\documentclass[10pt,conference,letterpaper]{IEEEtran}

\usepackage{cite}
\usepackage{latexsym}
\usepackage{amsmath}
\usepackage{amsthm}
\usepackage{amssymb}
\usepackage{amsfonts}
\usepackage{textcomp}
\usepackage{mathrsfs}
\usepackage[dvips,dvipdf]{graphicx}
\usepackage{algorithm}
\usepackage{algorithmic}
\usepackage{multirow}
\usepackage{bbm}
\usepackage{url}
\usepackage{pifont} 
\usepackage{enumitem}
\usepackage{dsfont}
\usepackage{color}
\usepackage{framed}
\usepackage{subfig}

\include{Complete.bib}

\newcommand{\wt}{\widetilde{\mathbf{w}}}

\newtheorem{thm}{Theorem}

\newtheorem{rem}{Remark}

\def\argmax{\operatornamewithlimits{arg\,max}}

\newcommand{\bE}{\mathds{E}}

\newcommand{\mb}{\mathbf}

\newcommand{\ol}{\overline}
\renewcommand{\wt}{\widetilde}
\newcommand{\id}{\mathds}
\newcommand{\wh}{\widehat}

\IEEEoverridecommandlockouts
    
\begin{document}

\title{Can We Achieve Fresh Information with Selfish Users in Mobile Crowd-Learning?}

\author{Bin Li$^{\dag}$ \mbox{\hspace{0.4cm}} Jia Liu$^{*}$
\\ $^{\dag}$Dept. of Electrical, Computer and Biomedical Engineering, University of Rhode Island
\\ $^{*}$Dept. of Computer Science, Iowa State University
\thanks{
Authors are listed in alphabetical order. Both authors are primary authors.
}
}

\maketitle


\begin{abstract}
The proliferation of smart mobile devices has spurred an explosive growth of mobile crowd-learning services, where service providers rely on the user community to voluntarily collect, report, and share real-time information for a collection of scattered points of interest. 
A critical factor affecting the future large-scale adoption of such mobile crowd-learning applications is the freshness of the crowd-learned information, which can be measured by a metric termed ``age-of-information'' (AoI). 
However, we show that the AoI of mobile crowd-learning could be arbitrarily bad under selfish users' behaviors if the system is poorly designed.
This motivates us to design efficient reward mechanisms to incentivize mobile users to report information in time, with the goal of keeping the AoI and congestion level of each PoI low.
Toward this end, we consider a simple linear AoI-based reward mechanism and analyze its AoI and congestion performances in terms of price of anarchy (PoA), which characterizes the degradation of the system efficiency due to 
selfish behavior of users. 
Remarkably, we show that the proposed mechanism achieves the optimal AoI performance asymptotically in a deterministic scenario.
Further, we prove that the proposed mechanism achieves a bounded PoA in general stochastic cases, and the bound only depends on system parameters. 
Particularly, when the service rates of PoIs are symmetric in stochastic cases, the achieved PoA is upper-bounded by $1/2$ asymptotically. 
Collectively, this work advances our understanding of information freshness in mobile crowd-learning systems.
\end{abstract}



\section{Introduction} \label{sec:intro}

Fueled by the proliferation of smart mobile devices (e.g., smartphones, tablets, etc.), recent years have witnessed a rapid growth of information services and data analytics based on large-scale {\em crowd-learning}.
A key defining feature of these crowd-learning applications is that they rely on the user community to voluntarily collect, report, and share real-time information for a set of distributed points of interest (PoI).
Such crowd-learned information will in turn benefit the users themselves and attract more users to join the community (by reputation, word of mouth, etc.), which further enhances the accuracy, value, and significance of the crowd-learning applications.
For example, the real-time traffic congestion and accident information on Google Waze\cite{Waze} (a community-based GPS system) relies on the reports from mobile devices and the tracking of their locations, densities, and trajectories.
As another example, by offering a variety of incentives, many data analytics services leverage their user communities to share real-time information of scattered commodities and resources, such as cheap gasoline stations (e.g., GasBuddy\cite{GasBuddy}), parking space availability (e.g., Pavemint\cite{Pavemint}), free WiFi hotspots (e.g., WiFi Finder\cite{WiFiFinder}), popular grocery deals information (e.g., Basket\cite{Basket}), to name just a few.
It can be foreseen that new crowd-learning applications will continue to emerge.

Although mobile crowd-learning holds a great potential to fundamentally change our modern society, 
a critical factor affecting its future large-scale adoption is the {\em freshness} of the crowd-learned information, which can be measured by a fundamental metric termed ``Age-of-Information'' (AoI).
Guaranteeing information freshness in crowd-learning is critical because stale information discourages existing and new users from participating, which in turn degrades the information freshness and creates a vicious circle.
Unfortunately, due to the special dynamics between the service provider and the users, there is an inherent lack of information freshness guarantee in mobile crowd-learning:
First, to maintain information freshness, the service provider needs to incentivize the users to update the states of the PoIs.
Second, the crowd-learning users are ``selfish'' in the sense that their best interest is to maximize their own benefit from participating in crowd-learning, rather than minimizing the AoI for the service provider.
Hence, a poorly designed incentive mechanism could result in two undesirable consequences: (i) too many users flock to an attractive PoI, which leads to redundant sampling and severe queueing congestion; and (ii) all other PoIs suffer from large AoI because of under-sampling.
In light of these unique characteristics of mobile crowd-learning, several fundamental open questions naturally arise:
\begin{list}{\labelitemi}{\leftmargin=1.2em \itemindent=-0.8em \itemsep=.2em}
\item[1)] Is it possible to guarantee information freshness by incentivizing selfish users in mobile crowd-learning?
\item[2)] If the answer to 1) is ``yes,'' what is the fundamental relationship between reward and AoI in crowd-learning?  
\item[3)] How to design reward mechanisms to avoid large queueing congestion while guaranteeing AoI in crowd-learning?
\end{list}


However, answering the above questions are non-trivial because the AoI and congestion analysis in mobile crowd-learning face the following challenges:
First, there is a lack of analytical model that characterizes the essential features of mobile crowd-learning in the literature.
Most of the existing work on crowd-sensing are based on static models that hardly capture the dynamic and stochastic nature of participating users in mobile crowd-learning.
Second, as shown by recent studies (see, e.g., \cite{Kaul12:AO_CISS,Yates12:AO_ISIT,Kaul12:AO_INFOCOM,Kaul11:AO_SECON}), AoI dynamics are fundamentally different from the traditional queueing evolution, which necessitates new theoretical tools.
Third, as will be shown later, there is a strong {\em coupling} between the AoI and queue-length processes in crowd-learning, where changing the design of either one would significantly affect that of the other.  

In this paper, we overcome the above challenges and propose a new analytical model coupled with the {\em Price of Anarchy} (PoA) metric, which characterizes the degradation of a system due to selfish behavior of users\footnote{The value of PoA is always between $0$ and $1$, and the larger the PoA, the less efficient the system. See Sections~\ref{sec:linear_mechanism}--\ref{sec:general} for more in-depth discussions.}. 
This enables us to analyze and understand the relationships between AoI, queueing congestion, and rewards under users' selfishness. 
The main results and contributions of this paper are as follows:

\begin{list}{\labelitemi}{\leftmargin=1em \itemindent=-0.5em \itemsep=.2em}
\item First, we develop a new analytical model for mobile crowd-learning, which takes into account the strong couplings between the stochastic arrivals of participating users, PoIs' information evolutions, and reward mechanisms.
As will be discussed next, this new analytical model enables us to reveal the fundamental scaling law between AoI, queueing congestion, and the reward rate set by the service provider. 


\item Next, as a starting point, we analyze the AoI performance under a linear AoI-based reward mechanism in a deterministic setting, where there is exactly one arriving user in each time slot, and each PoI serves exactly one user (if any) in each time slot (and hence no queueing effect in this setting).
We show that given an AoI reward rate $\beta$, the PoA is upper-bounded by $O(1/\beta)$, which implies that the system achieves the optimal AoI as $\beta$ increases asymptotically.


\item Finally, based on our results for the deterministic case, we characterize the joint AoI-congestion performance of mobile crowd-learning for stochastic settings.
Although the reward policy design for joint AoI and queueing congestion optimization remains an open problem in stochastic settings, surprisingly, we show that the above linear AoI-based reward mechanism yields a bounded PoA, which only depends on the arrival and service parameters of the system. In the case of symmetric services, the PoA  is upper-bounded by $1/2$ as the reward rate $\beta$ increases asymptotically. 


\end{list}

Collectively, our results in this paper advance the understanding of achieving information freshness in mobile crowd-learning with selfish users.
The remainder of this paper is organized as follows:
Section~\ref{sec:related_work} reviews related work.
Section~\ref{sec:model} introduces system model and problem statement.
Section~\ref{sec:linear_mechanism} introduces a linear reward mechanism, and
Sections~\ref{sec:periodic}--\ref{sec:general} study its PoAs in the deterministic and stochastic cases, respectively.
Section~\ref{sec:numerical} presents numerical results and Section~\ref{sec:conclusion} concludes this paper.


\section{Related Work} \label{sec:related_work}

To put our work in comparative perspectives, in this section, we provide an overview on the related work in the areas of crowd-sensing and age-of-information, respectively.

\smallskip
{\bf a) Crowd-Sensing:}
In the literature, crowd-sensing refers to the sensing model where a group of individuals collectively measure some common phenomena, e.g., environmental quality monitoring\cite{Creek_watch}, noise pollution assessment\cite{Maisonneuve09:noise_pollution,Rana10:noise_mapping}, and traffic monitoring\cite{Zhang14:crowdsensing}, etc. 
Although crowd-sensing bears some similarity to mobile crowd-learning, the main focuses of the crowd-sensing research community are on network resource management, system infrastructure, incentive mechanism designs, etc. (see \cite{Liu16:Crowdsensing_Survey} for a comprehensive survey).
In contrast, the overarching theme of this paper is to guarantee {\em information freshness} in learning {\em scattered} objects by a {\em selfish} crowd. 
Moreover, most of the existing crowd-sensing research adopts either a static model, where the set of sensing individuals is fixed (see, e.g., \cite{Chen16:search_system} and references therein); or based on a static game-theoretic model, where a fixed set of sensing individuals are incentivized/contracted by a fixed set of employers (see, e.g.,\cite{Han16:Crowdsensing_Game} and references therein).
These are fundamentally different from our {\em dynamic} model described in Section~\ref{sec:model}.
Hence, our work fills a critical gap in understanding large-scale mobile crowd-learning. 

\smallskip
{\bf b) Age-of-Information (AoI):}
Originated from sensing systems, AoI has attracted increasing attention from the information theory, signal processing, and communications communities in recent years.
Besides being a useful performance metric, AoI also possesses several key features that distinguish itself from the traditional notion of queueing delay.
Most notably, in many sensing systems, it has been found that while queueing delay benefits from lower sampling rates (implying less data traffic), AoI is non-monotone with respect to sampling rates.
This key difference has sparked AoI research in several aspects, e.g., real-time sampling and remote estimation trade-off\cite{Sun17:RT_ISIT,Gao15:RT_CDC}, joint source-channel coding exploitation\cite{Ceran17:SCC_ArXiv,Yates17:SCC_ISIT}, caching \cite{Kam17:CF_ISIT,Yates17:CF_ISIT}, optimization algorithms for AoI minimization \cite{Sun17:AO_TIT,Kaul17:AO_ISIT}, age-based scheduling\cite{Li15:AS_INFOCOM,Lakshminarayana09:AS_CHPCN}, just to name a few.
We note that the key differences between our research and the existing AoI research are: i) the tight coupling and dependence between {\em multi-user arrival dynamics} and {\em multi-source information time series} on a {\em network level}; and ii) the complex interactions between AoI, fresh/outdated information, and queueing, all of which are governed by the service provider's  {\em reward mechanism designs}.
These key differences introduce new challenges in guaranteeing stochastic network information freshness unseen in existing AoI research.


\section{Network Model and Problem Statement} \label{sec:model}

As shown in Fig.~\ref{fig:sys_model}, we consider a mobile crowd-learning system consisting of $N$ nodes that represent $N$ points of interest (PoI), e.g., road intersections, parking garages, potential WiFi hotspots, gas stations, etc.
We consider a time-slotted system.
In each time slot $t$, each PoI $n$ has some state information $p_{n}[t]$ (e.g., congestion level, parking rate and space, gas price, etc.) that is time-varying and to be sampled by their users. 
We assume that $p_n[t] \in [p_{\min}, p_{\max}], \forall t,$ for some positive constants $p_{\min}$ and $p_{\max}$.
A service provider (i.e., a crowd-learning-based information/data analytics platform) relies on randomly arriving users to sample and report the states of the PoIs.
The service provider maintains a record for each PoI, whose value in time slot $t$ is denoted as $r_{n}[t]$, $n=1,\ldots,N$. 
For ease of exposition, we will refer to $p_{n}[t]$ and $r_{n}[t]$ as ``\emph{price}'' and ``\emph{recorded price}'' in the rest of this paper, respectively.
Let $u_{n}[t]$ be the most recent update time up to time slot $t$ for PoI $n$'s record.
Hence, the {\em age} (freshness) of record $r_{n}[t]$ in time slot $t$ can be represented as $\Delta_{n}[t] = t-u_{n}[t]$.

\begin{figure}[t!]
\centering
\includegraphics[scale=0.66]{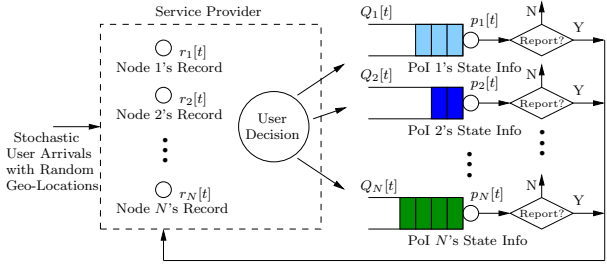}
\caption{A system model for mobile crowd-learning.}
\label{fig:sys_model}
\vspace{-0.2in}
\end{figure}

Let $A[t]$ be the number of users arriving at the system in time slot $t$. We assume that $A[t], t\geq0$, are  independently and identically distributed (i.i.d.) across time with mean $\lambda\triangleq \mathbb{E}[A[t]]>0$ and bounded second moment $\mathbb{E}[A^2[t]]<\infty$. 
The arrivals model the scenario that users at different locations use their mobile apps in each time slot to acquire information of the PoIs before making decisions. 
Each arriving user will first observe the current records of all PoIs and choose a favorable one (e.g., choosing the least congested route, the lowest gas price, or the cheapest and nearest parking space, etc.). 
However, due to the {\em random} updating time in crowd-learning, the information of some PoI $n$'s record could be old and hence $r_{n}[t]$ may be outdated and inaccurate. 


On the other hand, upon the arrival at his/her chosen PoI, say $n$ in time slot $t$, the user will report the PoI's real-time state (e.g., real-time price, congestion level, etc.), i.e., $p_{n}[t]$. Let $R_n[t]$ denote the number of users that can be served by PoI $n$ in time slot $t$. 
We assume that $R_n[t]$, $t\geq0$, are i.i.d. across time and independently distributed across PoIs with mean $\mu_n\triangleq\mathbb{E}[R_n[t]]>0, \forall n$, and $R_n[t]\leq R_{\max}, \forall n, t$, for some $R_{\max}<\infty$. 
We use $Q_n[t]$ to denote the number of users awaiting for service in PoI $n$ in time slot $t$.

The service provider's goal is to achieve minimum time-average AoI while keeping queueing congestion at each PoI low.
The rationale behind this goal is that low AoI (i.e., fresh information) implies multiple benefits, e.g., high information accuracy, which attracts more users;  hence more advertising revenues due to large user volume, etc.
However, the following toy example shows that the natural greedy behavior of selfish users could yield {\em AoI instability} in mobile crowd-learning:


\smallskip
\underline{{\em A Motivating Example (AoI Instability due to Selfishness):}}
Consider a two-PoI example as shown in Fig.~\ref{fig_2_User}.
Consider the most ``natural'' price-greedy decision made by selfish users:
In time slot $t$, each arriving user compares the recorded prices $r_{1}[t]$ and $r_{2}[t]$ and chooses the cheaper PoI, i.e., choosing $n^{*}[t] \in \arg\min_{n\in\{1,2\}} \{r_{n}[t]\}$.
Suppose that $p_{n}[t] \in [0, p_{\max}]$, $n=1,2$.
Assume that the probability $\mathrm{Pr}\{p_{n}[t] = p_{\max} \} = \epsilon$, $n=1,2$, where $\epsilon>0$ is some small value.
Suppose also that in the initial state, $p_{1}[0] = p_{\max}$ and $p_{2}[0] = \delta < p_{\max}$.
Thus, at $t=0$, {\em all} users choose PoI 2 and the record $r_{2}[t]$ will be updated, in which case the age of PoI 2 in time slot $1$ becomes zero, i.e., $\Delta_{2}[1]=0$.
However, due to the high initial price $p_{1}[0]$, no user chooses PoI 1.
Also, due to the low probability of $p_{2}[t]$ reaching $p_{\max}$, it would take a long time (could be unbounded if $\epsilon$ is arbitrarily small) for PoI 1 to receive {\em any} user to update $r_{1}[\cdot]$, although $p_{1}[t]$ may be lower than $p_{2}[t]$.
For example, in Fig.~\ref{fig_2U_0999}, $p_{1}[t]$ and $p_{2}[t]$ are uniformly distributed in $[0, 1]$.
We let $p_{1}[0] = 0.999$ (large initial value) and $p_{2}[0] = 0.1$. 
Clearly, we can see that PoI 1's AoI is large and grows linearly with respect to time.
\qed

\smallskip
The above observation of AoI instability due to users' selfishness motivates us to design crowd-learning reward mechanisms to ensure information freshness in crowd-learning.

\section{A Linear AoI-Based Reward Mechanism} \label{sec:linear_mechanism}
To keep the AoI being bounded, the service provider would like users to go to and sample a PoI with the most outdated information. 
However, unlike traditional scheduling problems, the crowd-learning service provider cannot enforce each arriving selfish user to go to a certain PoI. 
Rather, the service provider can only offer incentives/rewards to influence the users to choose certain PoIs.
So far, however, the problem of optimal reward mechanism design for mobile crowd-learning with selfish users has not been addressed in the literature.
Therefore, in this paper, we start from considering a simple {\em linear} reward mechanism for mobile crowd-learning.
 
 \begin{figure}[t!]
    \begin{minipage}[t]{0.56\linewidth}
    \vspace{-.99in}
        \includegraphics[width=1.02\textwidth]{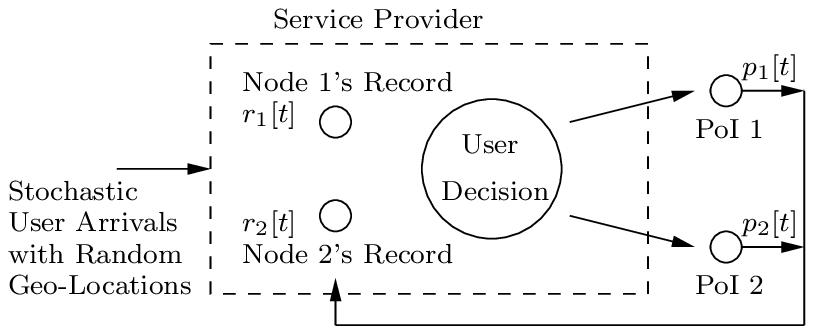}
        \vspace{-0.07in}
        \caption{{\small A two-PoI motivating example with $p_{1}[0] = 0.999$ and $p_{2}[0]=0.1$.}} \label{fig_2_User}
    \end{minipage}%
    \hspace{0.01\linewidth}
    \begin{minipage}[t]{0.4\linewidth}
        \centering
        \includegraphics[width=1\textwidth]{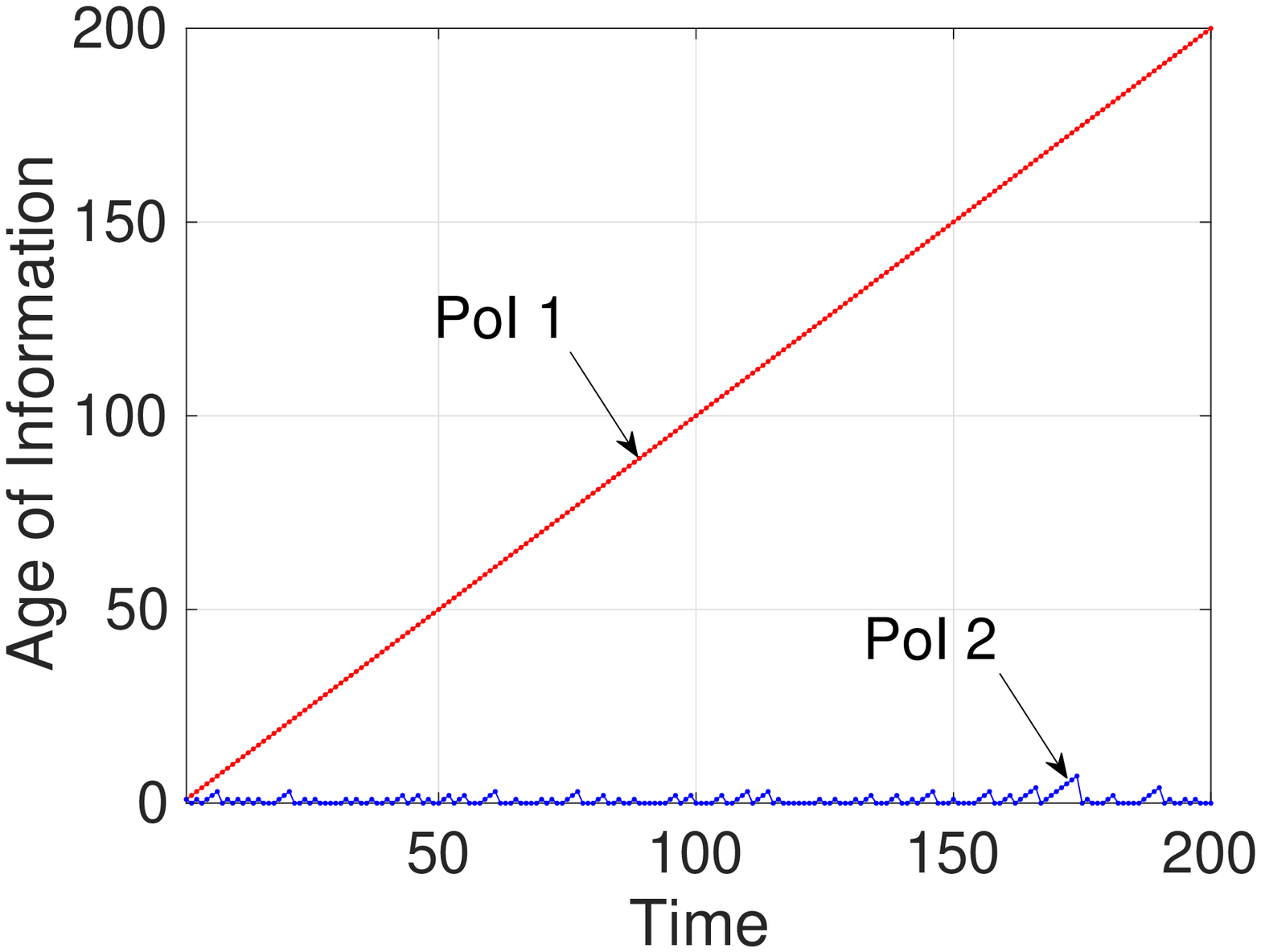}
        \vspace{-.25in}
        \caption{{\small Large and unstable AoI of PoI 1 in Fig.~\ref{fig_2_User}.}}\label{fig_2U_0999}
    \end{minipage}%
\vspace{-.15in}
\end{figure}

Specifically, we let $\beta>0$ represent the ``\emph{reward per unit of age}'' offered by the service provider. Note that each user prefers to select a PoI with both low price and congestion level. We use a parameter $\gamma>0$ to denote users' sensitivity to queueing congestion, which depends on specific mobile crowd-learning application\footnote{Here, we assume that all users are homogeneous and have the same $\gamma$-value. The impact of users' heterogeneity in congestion sensitivity will be studied in our future work.}. Hence, in each time slot $t$, an arriving user's {\em presumed benefit} for choosing PoI $n$ and reporting its state is: $\beta\Delta_n[t]-\gamma Q_n[t]-r_n[t]$.  
In this work, we assume that all arriving users are selfish and rational, so that they would select a PoI $n^*[t]$ to maximize their presumed benefit, i.e.,
\begin{align}
\label{eqn:user:behavior:general}
n^*[t]\in\argmax_{n \in \{ 1,2,\ldots,N \} } \left(\beta\Delta_n[t] - \gamma Q_n[t] - r_n[t]\right), \,\, \forall t.
\end{align}
Note that for any fixed $\gamma$, when the reward rate diminishes, i.e., $\beta\downarrow0$, each user essentially follows the ``greedy" scheme to select a PoI with the smallest value of $\gamma Q_n[t]+r_n[t]$. 
By contrast, when the reward rate approaches infinity, i.e., $\beta\uparrow\infty$, the effects of $Q_n[t]$ and $r_n[t]$ become negligible in users' presumed benefit and thus it encourages users to help the service provider maintain information freshness. 


%
To facilitate our subsequent analysis, we use $S_n[t]$ to denote whether there is at least one user selecting PoI $n$ in time slot  $t$. In particular, $S_{n}[t]=1$ if at least one user selects PoI $n$ in time slot $t$, and $S_{n}[t]=0$ otherwise. Under the assumption of users' selfishness and rationality, the dynamics of queue-length and age of PoI $n$ can be described as follows:
\begin{align}
\label{eqn:Queue}
Q_n[t+1]=\max\{Q_n[t]+A[t]S_n^{*}[t]-R_n[t],0\}, \forall n.
\end{align}
\begin{equation}
\label{eqn:Delta}
\text{and }\Delta_n[t+1] =
  \begin{cases}
       \Delta_n[t]+1, & \text{if $S_{n}^*[t]\id{1}_{\{A[t]>0\}}=0$}, \\
        0,            & \text{otherwise}, 
  \end{cases}
\end{equation}
where $S_n^{*}[t]=1$ if $n=n^*[t]$ and $S_n^{*}[t]=0$ otherwise, and $\id{1}_{\{\cdot\}}$ is an indicator function. Let $\mb{S}^*[t]\triangleq (S_n^*[t])_{n=1}^{N}$.

To understand the impact of users' selfishness on AoI and queueing congestion, in this paper, we adopt the so-called {\em Price of Anarchy} (PoA) metric from the game theory literature, which characterizes  the degradation of the system efficiency due to the selfish behavior of users compared to the optimum. 
Roughly speaking, the notion of PoA $\rho$ is defined as: 
\begin{align} \label{eqn:def:PoA}
\rho = 1 - \frac{\text{Minimum cost}}{\text{Cost under selfish behavior}}. 
\end{align}
Note that $\rho\in[0,1]$ and the smaller the PoA, the more efficient the system under selfish user behavior. 
In what follows, we will analyze the PoA of the linear reward scheme (\ref{eqn:user:behavior:general}), where the definition of cost in (\ref{eqn:def:PoA}) depends on specific system scenarios that will be clarified in subsequent sections.





\section{Price of Anarchy: A Deterministic Case} \label{sec:periodic}

In this section, we first consider a simple deterministic case, where, in each time slot, there is exactly one arriving user and each PoI serves exactly one user if there is any. 
This deterministic case not only provides interesting insights, its results and proof strategies will also serve as a foundation for analyzing general cases with stochastic arrivals and services.
Note that due to the special arrival and service patterns in this deterministic case, there is no queueing effect at each PoI.
Hence, user's selfish selection (cf. \eqref{eqn:user:behavior:general}) becomes:
\begin{align}
\label{eqn:user:behavior:periodic}
n^*[t]\in\argmax_{n\in\{1,2,\ldots,N\}} \left(\beta\Delta_n[t] - r_n[t]\right).
\end{align}
In addition, the evolution of age of PoI $n$ in (\ref{eqn:Delta}) becomes:
\begin{align} \label{eqn:Delta1}
\Delta_n[t+1]=(\Delta_n[t]+1)(1-S_n^*[t]).
\end{align}

Next, we study the information freshness performance under the selfish behavior of users (cf. \eqref{eqn:user:behavior:periodic}) based on the notion of PoA.
Since queueing does not play a role and the system is symmetric, we define the cost function with selfish users under some reward rate $\beta$ as $\ol{\Delta}_{\max}^{(\beta)}$ and hence the PoA is $\rho(\beta)\triangleq 1 - \ol{\Delta}_{\max}^{(\text{OPT})}/\ol{\Delta}_{\max}^{(\beta)}$, where $\ol{\Delta}_{\max}^{(\text{OPT})}$ and $\ol{\Delta}_{\max}^{(\beta)}$ are the average maximum age under an optimal policy (with unselfish users) and under the user's selfishness, respectively. 
The first main result of this paper is stated as follows:

\begin{thm}[AoI-Based PoA for the Deterministic Case]  
\label{thm:periodic}
If there is exactly one user arriving in each time slot and each PoI serves exactly one user per time-slot if there is any, the users' selfishness yields the following PoA performance: 
\begin{align}
\rho(\beta)\leq\frac{p_{\max}}{(N-1)\beta+p_{\max}} = O(1/\beta).
\end{align}
\end{thm}
\begin{proof} 
The proof consists of two main steps: (i) Finding an upper bound on the average maximum age $\ol{\Delta}_{\max}^{(\beta)}$ due to users' selfishness; 
and (ii) derving a lower bound on the average maximum age $\ol{\Delta}_{\max}^{\text{(OPT)}}$ achieved by an optimal policy.

\textbf{Step 1)}: To find an upper bound on the average maximum age due to users' selfish behavior, we perform Lyapunov drift analysis through an age-based Lyapunov function as follows:
\begin{align}
\label{eqn:Lyapunov:periodic}
V[t] \triangleq \sum_{n=1}^{N}\Delta_{n}[t].
\end{align}

Let $\mb{M}[t] \!\triangleq\! (\{\Delta_n[t]\}_{n=1}^{N}, \{r_n[t]\}_{n=1}^{N})$ and consider an unselfish policy $\wt{\mb{S}}[t] \!\triangleq\! (\wt{S}_n[t])_{n=1}^{N} \!\in\! \argmax_{\mb{S}} \sum_{n=1}^{N} \Delta_n[t] S_n[t]$, i.e., users select the PoI with the largest age.
Then, the one-step conditional expected drift of $V[t]$ can be computed as:
\begin{align}
\Delta V[t] \triangleq & \bE\left[V[t+1]-V[t]\middle|\mb{M}[t]\right] \nonumber\\
=&\sum_{n=1}^{N}\bE\left[\Delta_n[t+1]-\Delta_n[t]\middle|\mb{M}[t]\right]  \nonumber\\
\stackrel{(a)}{=}&\sum_{n=1}^{N}\bE\left[1-(\Delta_n[t]+1)S_n^{*}[t]\middle|\mb{M}[t]\right] \nonumber\\
\stackrel{(b)}{=}& N-1 -\sum_{n=1}^{N}\bE\left[\Delta_n[t]S_n^{*}[t]\middle|\mb{M}[t]\right] \label{eqn:Lyapunov:periodic:drift}\\
\leq&N-1 -\sum_{n=1}^{N} \bE\left[\left(\Delta_n[t]-\frac{1}{\beta}r_n[t]\right)S_n^{*}[t] \middle|\mb{M}[t]\right]\nonumber\\
\stackrel{(c)}{\leq}& N-1 -\sum_{n=1}^{N} \bE\left[\left(\Delta_n[t]-\frac{1}{\beta}r_n[t]\right)\wt{S}_n[t] \middle|\mb{M}[t]\right], \nonumber\\
\stackrel{(d)}{\leq}& N-1-\Delta_{\max}[t]+\frac{1}{\beta}p_{\max},\label{eqn:Drift_d_periodic} 
\end{align}
where $(a)$ uses dynamics of $\Delta_n[t]$ in (\ref{eqn:Delta1}); $(b)$ follows from the fact that each user joins one of the PoIs in each time slot, i.e., $\sum_{n=1}^{N}S_n^{*}[t]=1$; $(c)$ follows from the definition of $S_{n}^{*}[t]$;
and $(d)$ uses the fact that $r_n[t]\leq p_{\max},\forall n, t\geq0$, the definition of $\wt{\mb{S}}[t]$, and the fact that exactly one $\wt{S}_n[t]$ is non-zero.  
It then follows from (\ref{eqn:Drift_d_periodic}) that:
\begin{align} \label{eqn:Drift_ucexp}
\bE\left[V[t+1]-V[t]\right]\leq N-1-\bE[\Delta_{\max}[t]]+\frac{1}{\beta}p_{\max}.
\end{align}
Summing (\ref{eqn:Drift_ucexp}) for $t=0,1,2,\ldots, T-1$, we obtain:
\begin{align*}
\bE[V[T]-V[0]] \!\leq\! -\sum_{t=0}^{T-1}\bE[\Delta_{\max}[t]] \!+\! (N-1)T \!+\!\frac{T}{\beta}p_{\max},
\end{align*}
which implies that
\begin{align}
\label{eqn:prop:ub}
\ol{\Delta}_{\max}^{(\beta)}\triangleq\limsup_{T\rightarrow\infty}\frac{1}{T}\sum_{t=0}^{T-1}\bE[\Delta_{\max}[t]]\leq N-1+\frac{1}{\beta}p_{\max}. \!\!\!
\end{align}
\textbf{Step 2)}: Next, we derive a fundamental lower bound on the average maximum age that can be achieved by the optimal policy. 
By using the same Lyapunov function in (\ref{eqn:Lyapunov:periodic}) to compute the conditional expected one-step drift under the optimal policy $\{S_{n}^{(\text{OPT})}[t] \}$ and following similar steps, we have 
$\Delta V[t] = N-1-\sum_{n=1}^{N}\bE[\Delta_n^{(\text{OPT})}[t]S_n^{(\text{OPT})}[t]|\mb{M}[t]]$,
where $\Delta_n^{(\text{OPT})}[t]$ is the age of PoI $n$ in time slot $t$ under the optimal policy. In Step 1, we have already shown that the average maximum age is finite under the selfish policy. This readily implies that the average maximum age is also finite under the optimal policy. Therefore, $\bE[\Delta V[t]]$ will be equal to zero in steady-state and thus we have 
\begin{align}
\label{eqn:prop:periodic:lb}
\sum_{n=1}^{N}\bE[\wh{\Delta}_n^{(\text{OPT})}\wh{S}_n^{(\text{OPT})}]=N-1,
\end{align}
where $\wh{\Delta}_n^{(\text{OPT})}$ and $\wh{S}_n^{(\text{OPT})}$ are random variables with the same distribution as $\Delta_n^{(\text{OPT})}[t]$ and $S_n^{(\text{OPT})}[t]$ in steady-state under the optimal policy, respectively. Hence, we have 
\begin{align} \label{eqn:Age_LB_final}
\ol{\Delta}_{\max}^{(\text{OPT})}&\stackrel{(a)}{=}\bE[\wh{\Delta}_{\max}^{(\text{OPT})}] \stackrel{(b)}{=}\bE[\max_{n}\wh{\Delta}_n^{(\text{OPT})}]\nonumber\\
&\stackrel{(c)}{\geq}\sum_{n=1}^{N}\bE[\wh{\Delta}_n^{(\text{OPT})}\wh{S}_n^{(\text{OPT})}]\stackrel{(d)}{=}N-1,
\end{align}
where step $(a)$ follows from the boundedness of the average maximum age under the optimal policy; $(b)$ is true for $\wh{\Delta}_{\max}^{(\text{OPT})}\triangleq\max_n\wh{\Delta}_n^{(\text{OPT})}$; $(c)$ follows from the fact that each arriving user joins exactly one of the PoIs, i.e., $\sum_{n=1}^{N}\wh{S}_n^{(\text{OPT})}=1$; and $(d)$ uses \eqref{eqn:prop:periodic:lb}.
Lastly, by combining the upper bound in Step 1 and lower bound in Step 2, the desired PoA result in Theorem~\ref{thm:periodic} follows and the proof is complete.
\end{proof}

\begin{rem}{\em
Two insightful remarks for Theorem~\ref{thm:periodic} are in order:
i) In Step 2, the lower bound of $\Delta_{\max}^{\text{(OPT)}}$ is {\em tight} and can be achieved by the Round-Robin policy, i.e., the system guides each arriving user to the PoIs in a Round-Robin fashion. 
Indeed, under Round-Robin, the ages of PoIs are a permutation of $\{0,1,2,\ldots,N-1\}$ in each time slot, and hence the maximum age under Round-Robin is $\Delta_{\max}^{\text{(RR)}}[t]=N-1, \forall t\geq0$, which implies that $\ol{\Delta}_{\max}^{\text{(RR)}}=N-1 = \ol{\Delta}_{\max}^{\text{(OPT)}}$; 
ii) From Theorem~\ref{thm:periodic}, we can observe that if $\beta$ increases asymptotically (i.e., $\beta\uparrow\infty$), we have $\rho(\beta)\downarrow 0$. 
This implies that the system is optimal and mimicking Round-Robin when the service provider increases the incentive asymptotically.
On the other hand, if $\beta$ reduces to zero (i.e., $\beta\downarrow0$), we can see from (\ref{eqn:user:behavior:periodic}) that each user just follows a price-greedy strategy.
In this case, Theorem~\ref{thm:periodic} suggests that the upper bound of $\rho(\beta)$ approaches $1$, which is consistent with our observation (cf. motivating example in Section~\ref{sec:model}) that the system suffers a poor AoI performance and potentially AoI instability (i.e., $\ol{\Delta}_{\max}^{(\beta)} \uparrow \infty$).}
\end{rem}


\section{Price of Anarchy: Stochastic Cases} \label{sec:general}

Based on the results for the deterministic case, we are now in a position to analyze the AoI and congestion performances under users' selfishness in cases with stochastic arrivals and services.
To facilitate analysis, we define a parameter $q\!\triangleq\!\mathrm{Pr}\{A[t]\!>\!0\}$ for the arrivals, which is strictly positive for $\lambda\triangleq\bE[A[t]]>0$. Let $\mu_{\Sigma}\triangleq\sum_{n=1}^{N}\mu_n$. 
Here, we adopt the cost function $J(\beta,\gamma) \triangleq \frac{\gamma\epsilon}{N}\sum_{n=1}^{N}\ol{Q}_n + \beta\sum_{n=1}^{N} \frac{\mu_{n}}{\mu_{\Sigma}} \ol{\Delta}_{n}$, where $\epsilon>0$ satisfies $\mu_n/\lambda\geq\mu_n/\mu_{\Sigma}+\epsilon/N, \forall n=1,2,\ldots,N$ due to the fact that $\lambda<\mu_{\Sigma}$ (necessary for guaranteeing the system's queueing stability\footnote{In this paper, we say that a queue $n$ is \emph{stable} if its average queue-length is finite, i.e., $\limsup_{T\rightarrow\infty}\frac{1}{T} \sum_{t=0}^{T-1} \bE[Q_n[t]]<\infty$. A system is \emph{stable} if all its queues are stable.}), and $\ol{Q}_n$ and $\ol{\Delta}_n$ are the average queue-length and average age of PoI $n$ under the user's selfishness, respectively. We note that in $J(\beta,\gamma)$, $\epsilon$ is used as a scaling parameter to reduce the cost's sensitivity to average queue-length $\frac{1}{N}\sum_{n=1}^{N}\ol{Q}_n$ under different arrival rates $\lambda$.
Also, $\gamma$ and $\beta$ are used to emphasize the relative importance between queueing and AoI costs, as in the presumed benefit for users' selfish decisions (cf. \eqref{eqn:user:behavior:general}).
Also, note that $J(\beta,\gamma)$ is based on weighted average age, where the weight $\frac{\mu_{n}}{\mu_{\Sigma}}$ is used to ``equalize'' the different AoI scales caused by the heterogeneity of the PoIs\footnote{This is also motivated by the fact that the service provider prefers a better AoI for the PoI with a faster service rate.}. 
As a result, the PoA is specialized to $\rho(\beta,\gamma)\triangleq 1-\frac{J^{\text{(OPT)}}(\beta,\gamma)}{J(\beta,\gamma)}$.
Our second key result is for the stochastic cases and stated as follows:


\begin{thm}[Joint AoI-Congestion PoA of Stochastic Cases]
\label{thm:general:asym}
If $\lambda<\mu_{\Sigma}$, then there exists an $\epsilon >0$ satisfying $\mu_n/\lambda\geq\mu_n/\mu_{\Sigma}+\epsilon/N, \forall n=1,2,\ldots,N$. In such a case, the users' selfishness yields the following PoA performance:
\begin{multline}
\rho(\beta,\gamma)\leq \frac{B(\gamma)-\gamma M + p_{\max}}{B(\gamma)+\beta\left(\frac{N}{q}-1\right)+ p_{\max}}\\
+\frac{\beta\left(\frac{N}{q}-\frac{1}{2q\mu_{\max}}\sum_{n=1}^{N}\mu_n-\frac{1}{2}\right)}{B(\gamma)+\beta\left(\frac{N}{q}-1\right)+p_{\max}},
\end{multline}
where $B(\gamma) \triangleq \frac{\gamma}{2\lambda}\big(\bE[A^2[t]]+\sum_{n=1}^{N}\bE[R_n^2[t]]\big)$, $M\triangleq\frac{\epsilon}{2N(\mu_{\Sigma}-\lambda)}\big(\text{Var}(A[t])+\sum_{n=1}^{N}\text{Var}(R_n[t])+(\mu_{\Sigma}-\lambda)^2\big)-\frac{1}{2}\epsilon R_{\max}$, and $\mu_{\max}\triangleq\max_{n}\mu_n$.
\end{thm}

\begin{proof}
Similar to the proof of Theorem~\ref{thm:periodic}, we first find an upper bound on $J(\beta,\gamma)$ by using the Lyapunov drift analysis and then determine a fundamental lower bound on $J^{\text{(OPT)}}(\beta,\gamma)$.

\textbf{Step 1)}: Consider the following Lyapunov function: $
L[t]\triangleq \frac{\gamma}{2\lambda\beta}\sum_{n=1}^{N}Q_n^2[t] + \frac{1}{q}\sum_{n=1}^{N}\Delta_n[t]$.
Let $\mb{Z}[t]\triangleq ((Q_n[t])_{n=1}^{N}, (\Delta_n[t])_{n=1}^{N},(r_n[t])_{n=1}^{N})$.
Then, the one-step conditional expected drift can be computed as:
\begin{small}
\begin{align}
&\Delta L[t] \triangleq \bE\left[L[t+1]-L[t]\middle|\mb{Z}[t]\right] = \! \bE \bigg[\frac{\gamma}{2\lambda\beta}\sum_{n=1}^{N} \big(Q_n^2[t + 1] \nonumber\\
& \qquad - Q_n^2[t]\big) + \frac{1}{q} \sum_{n=1}^{N} \! (\Delta_n[t + 1] \!-\!\Delta_n[t] ) \bigg|\mb{Z}[t]\bigg]  \nonumber\\
&\stackrel{(a)}{\leq} \frac{B(\gamma)}{\beta} + \bE\bigg[\frac{\gamma}{\lambda\beta}\sum_{n=1}^{N}Q_n[t](A[t]S_n^{*}[t]-R_n[t]), \nonumber\\
\label{eqn:one-step_Lyapunov_sym} &\qquad+\frac{1}{q}\sum_{n=1}^{N}\left(1-(\Delta_n[t]+1)S_n^{*}[t]\id{1}_{\{A[t]>0\}}\right)\bigg|\mb{Z}[t]\bigg],
\end{align}
\end{small}
where $(a)$ is true for $B(\gamma)=\frac{\gamma}{2\lambda}(\bE[A^2[t]]+N R_{\max}^2)<\infty$, and uses dynamics of $Q_n[t]$ (cf. \eqref{eqn:Queue}) and $\Delta_n[t]$ (cf. \eqref{eqn:Delta}), and the fact that $(\max\{x,0\})^2\leq x^2, \forall x$.

Next, we let $\mb{Z}'[t]\triangleq(\mb{Z}[t],\id{1}_{\{A[t]>0\}})$.
Then, for any function $f(\mb{Z}[t])$, the following sequence of equalities holds:
\begin{align}
\bE\{f(\mb{Z}[t])A[t]|\mb{Z}[t]\}=&q\bE\{f(\mb{Z}[t])A[t]|\mb{Z}'[t]\} \nonumber\\
=&q\bE\{A[t]|A[t]>0\} \bE\{f(\mb{Z}[t])|\mb{Z}'[t]\} \nonumber\\
\label{eqn_condexp_eq} =&\bE\{A[t]\} \bE\{ f(\mb{Z}[t])|\mb{Z}'[t]\}.
\end{align}
Note that each arriving user joins one of the PoIs in each time slot, i.e., $\sum_{n=1}^{N}S_n^{*}[t]=1$.
Also, the users' decisions $\mb{S}^*[t]$ only depend on $\mb{Z}[t]$.
Hence, we have that:
\begin{align}
\!\!\!\!\! (\ref{eqn:one-step_Lyapunov_sym}) \stackrel{(a)}{\leq}&\frac{B(\gamma)}{\beta} + \frac{N}{q} -1 -\frac{\gamma}{\lambda\beta}\sum_{n=1}^{N} \mu_{n} Q_n[t] \nonumber\\
\quad+&\bE \bigg[\frac{\gamma}{\beta}\sum_{n=1}^{N}Q_n[t]S_n^{*}[t]-\sum_{n=1}^{N}\Delta_n[t]S_n^{*}[t]\bigg|\mb{Z}'[t] \bigg] \nonumber\\
\leq& \frac{B(\gamma)}{\beta}+\frac{N}{q}-1 -\frac{\gamma}{\lambda\beta}\sum_{n=1}^{N}\mu_nQ_n[t]   \nonumber\\
\label{eqn:one-step_Lyapunov_sym1} \quad-&\bE \bigg[\sum_{n=1}^{N}\left(\Delta_n[t]-\frac{\gamma}{\beta} Q_n[t]-\frac{1}{\beta}r_n[t]\right)S_n^{*}[t]\bigg|\mb{Z}'[t] \bigg], \!\!\!
\end{align}
where $(a)$ follows from (\ref{eqn_condexp_eq}) and the fact that $q\in (0,1)$.
Next, consider an unselfish stationary randomized policy with $\bE[\wt{S}_n[t]]=\mu_n/\mu_{\Sigma}$, $\forall n$, if $A[t]>0$, and $\mu_{\Sigma}\triangleq\sum_{n=1}^{N}\mu_n$.
Clearly, from the definition of $\mb{S}^{*}[t]$, we have:
\begin{align}
&(\ref{eqn:one-step_Lyapunov_sym1}) \leq \frac{B(\gamma)}{\beta}+\frac{N}{q}-1-\frac{\gamma}{\beta}\sum_{n=1}^{N} \frac{\mu_{n}}{\lambda} Q_n[t] \nonumber\\
&\label{eqn:one-step_Lyapunov_sym2} -\bE \bigg[\sum_{n=1}^{N}\left(\Delta_n[t]-\frac{\gamma}{\beta} Q_n[t]-\frac{1}{\beta}r_n[t]\right)\wt{S}_n[t]\bigg|\mb{Z}'[t]\bigg],
\end{align}
Noting that $\mu_n/\lambda\geq \mu_{n}/\mu_{\Sigma} + \epsilon/N, \forall n$, we have
\begin{align*}
&\Delta L[t]\stackrel{(a)}{\leq} \frac{B(\gamma)}{\beta}+\frac{N}{q}-1+\frac{p_{\max}}{\beta}-\frac{\gamma\epsilon}{N\beta}\sum_{n=1}^{N} Q_n[t]\nonumber\\
&-\frac{\gamma}{\beta}\sum_{n=1}^{N}\frac{\mu_n}{\mu_{\Sigma}}Q_n[t]-\sum_{n=1}^{N}\frac{\mu_n}{\mu_{\Sigma}}\bE\left[\left(\Delta_n[t]-\frac{\gamma}{\beta} Q_n[t]\right)\middle|\mb{Z}'[t]\right] \nonumber\\
=&\!-\! \frac{\gamma\epsilon}{N\beta}\sum_{n=1}^{N} Q_n[t] \!-\! \sum_{n=1}^{N}\frac{\mu_n}{\mu_{\Sigma}}\Delta_n[t] \!+\! \frac{B(\gamma)}{\beta} \!+\! \frac{N}{q} \!-\!1 \!+\! \frac{p_{\max}}{\beta},
\end{align*}
where $(a)$ follows from $r_n[t]\leq p_{\max},\forall n,t$, and the definition of the stationary randomized policy $\{\wt{\mb{S}}[t]\}_{t\geq0}$. 
This implies 
\begin{multline} \label{eqn:uc_exp_asym}
\bE\left[L[t+1]-L[t]\right] \leq -\frac{\gamma\epsilon}{N\beta}\sum_{n=1}^{N}\bE[Q_n[t]]\\
- \sum_{n=1}^{N} \frac{\mu_n}{\mu_{\Sigma}} \bE[\Delta_{n}[t]] +\frac{B(\gamma)}{\beta}+\frac{N}{q}-1+\frac{p_{\max}}{\beta}.
\end{multline}
Summing (\ref{eqn:uc_exp_asym}) for $t=0,1,2,\ldots, T-1$, we obtain 
\begin{align*}
&\bE[L[T]-L[0]]\leq-\frac{\gamma\epsilon}{N\beta} \sum_{t=0}^{T-1}\sum_{n=1}^{N}\bE[Q_n[t]]\nonumber\\
&- \sum_{t=0}^{T-1}\sum_{n=1}^{N} \frac{\mu_n}{\mu_{\Sigma}} \bE[\Delta_{n}[t]]+\left(\frac{B(\gamma)}{\beta}+\frac{N}{q}-1+\frac{p_{\max}}{\beta}\right)T,
\end{align*}
which further implies the following upper bound on $J(\beta,\gamma)$:
\begin{align}
\label{eqn:prop:stochastoc:ub}
J(\beta,\gamma) \! \triangleq\!&\limsup_{T\rightarrow\infty}\frac{1}{T} \!\!\sum_{t=0}^{T-1} \!\! \left[\frac{\gamma\epsilon}{N}\sum_{n=1}^{N} \bE[Q_n[t]] \!+\! \beta\sum_{n=1}^{N} \frac{\mu_n}{\mu_{\Sigma}} \bE[\Delta_{n}[t]]\right] \nonumber\\
\leq&B(\gamma)+\beta\left(\frac{N}{q}-1\right)+p_{\max}.
\end{align}

\textbf{Step 2)}: Next, we derive a fundamental lower bound on $J^{\text{(OPT)}}(\beta,\gamma)$. Since we have shown that $J(\beta,\gamma)$ is upper-bounded under the selfish policy in Step 1, $J^{\text{(OPT)}}(\beta,\gamma)$ is also bounded under the optimal policy. Therefore, we have $J^{\text{(OPT)}}(\beta,\gamma)=\frac{\gamma\epsilon}{N}\sum_{n=1}^{N}\bE[\wh{Q}_n^{\text{(OPT)}}]+\beta\sum_{n=1}^{N}\frac{\mu_n}{\mu_{\Sigma}}\bE[\wh{\Delta}_n^{\text{(OPT)}}]$, where 
$\wh{Q}_n^{\text{(OPT)}}$ and $\wh{\Delta}_n^{\text{(OPT)}}$ are random variables with the same distribution as $Q_n[t]$ and $\Delta_n[t]$ in steady-state under the optimal policy, respectively.  Next, we lower-bound $\sum_{n=1}^{N}\bE[\wh{Q}_n^{\text{(OPT)}}]$ and $\sum_{n=1}^{N}\mu_n\bE[\wh{\Delta}_n^{\text{(OPT)}}]$ individually. 
In the rest of the proof, we omit the signifier ``$(\text{OPT})$'' for notational convenience and better readability. 

We first consider $\sum_{n=1}^{N}\mu_n\bE[\wh{\Delta}_n]$. 
By choosing the Lyapunov function $V_1[t]\triangleq\sum_{n=1}^{N}\mu_n\Delta_n[t]$ and following similar steps as in the derivation of \eqref{eqn:Lyapunov:periodic:drift}, we have 
\begin{align*}
&\!\! \Delta V_1[t] \!\triangleq\! \bE\left[V_1[t+1] \!-\! V_1[t]\middle|\mb{Z}'[t]\right]=\mu_{\Sigma}-\nonumber\\
&\qquad\quad q\sum_{n=1}^{N}\! \mu_n\bE\{S_n[t]|\mb{Z}'[t]\} \!-\! q\!\sum_{n=1}^{N} \!  \mu_n\bE\{\Delta_n[t]S_n[t]|\mb{Z}'[t]\}.
\end{align*}
Since $J^{(\text{OPT})}(\beta,\gamma)$ is bounded under the optimal policy, the weighted sum of average age must also be finite under the optimal policy. Therefore, one can conclude that $\bE[\Delta V_1[t]]=0$ in steady-state.
It then follows that:
\begin{align}
\label{eqn:prop:general:asym:lb:first}
\sum_{n=1}^{N}\mu_n\bE\left[\wh{\Delta}_n\wh{S}_n\right]=\frac{1}{q}\mu_{\Sigma}-\sum_{n=1}^{N}\mu_n\bE[\wh{S}_n],
\end{align}  
where $\wh{S}_n$ is the random variable with the same distribution as $S_n[t]$ in the steady-state under the optimal policy. 

Similarly, using Lyapunov function $V_2[t] \!\!\triangleq\!\! \sum_{n=1}^{N}\!\mu_n \Delta_n^2[t]$ and setting its drift to zero in steady-state yields:
\begin{align}
\label{eqn:prop:general:asym:lb:second}
2\sum_{n=1}^{N}\mu_n\bE[\wh{\Delta}_n] \!=\! q\sum_{n=1}^{N}\mu_n\bE[\wh{\Delta}_n^2\wh{S}_n] \!+\! q\sum_{n=1}^{N}\mu_n\bE[\wh{\Delta}_n\wh{S}_n]. \!\!\!
\end{align}
For any sample path, by Cauchy-Schwarz's Inequality, we have 
\begin{align}
\bigg(\sum_{n=1}^{N}\mu_n\wh{\Delta}_n\wh{S}_n\bigg)^2=&\bigg(\sum_{n=1}^{N}\sqrt{\mu_n\wh{S}_n}\cdot\sqrt{\mu_n\wh{S}_n}\wh{\Delta}_n\bigg)^2 \nonumber\\
\leq&\bigg(\sum_{n=1}^{N}\mu_n\wh{S}_n\bigg) \bigg(\sum_{n=1}^{N}\mu_n\wh{\Delta}_n^2\wh{S}_n\bigg),
\end{align}
which implies $\sum_{n=1}^{N}\mu_n\wh{\Delta}_n^2\wh{S}_n\geq\frac{(\sum_{n=1}^{N}\mu_n\wh{\Delta}_n\wh{S}_n)^2}{\sum_{n=1}^{N}\mu_n\wh{S}_n}$, and hence
\begin{align}
\label{eqn:prop:general:asym:lb:temp}
\bE\left[\sum_{n=1}^{N}\mu_n\wh{\Delta}_n^2\wh{S}_n\right]\geq \bE\Bigg[\frac{\left(\sum_{n=1}^{N}\mu_n\wh{\Delta}_n\wh{S}_n\right)^2}{\sum_{n=1}^{N}\mu_n\wh{S}_n}\Bigg].
\end{align}
Since $f(X,Y)=X^2/Y$ is convex for all $X \geq 0$ and $Y>0$, by using Jensen's Inequality, we have $\bE[\frac{X^2}{Y}] \geq \frac{(\bE[X])^2}{\bE[Y]}$.
Thus, setting $X=\sum_{n=1}^{N}\mu_n\wh{\Delta}_n\wh{S}_n$ and $Y=\sum_{n=1}^{N}\mu_n\wh{S}_n$, inequality \eqref{eqn:prop:general:asym:lb:temp} becomes:
\begin{align}
\label{eqn:prop:general:asym:lb:third}
\sum_{n=1}^{N}\mu_n\bE\left[\wh{\Delta}_n^2\wh{S}_n\right]\geq\frac{\left(\sum_{n=1}^{N}\mu_n\bE[\wh{\Delta}_n\wh{S}_n]\right)^2}{\sum_{n=1}^{N}\mu_n\bE[\wh{S}_n]}.
\end{align}
By combining \eqref{eqn:prop:general:asym:lb:first}, \eqref{eqn:prop:general:asym:lb:second} and \eqref{eqn:prop:general:asym:lb:third}, we have: 
\begin{align} \label{eqn:AgeLB_asym}
\sum_{n=1}^{N} \mu_n\bE[\wh{\Delta}_n] \geq& \frac{\mu_{\Sigma}}{2} \bigg[ \frac{\mu_{\Sigma}}{q\sum_{n=1}^{N} \mu_n\bE[\wh{S}_n]} -1\bigg]\nonumber\\
\geq& \frac{\mu_{\Sigma}}{2}\!\bigg[\frac{\mu_{\Sigma}}{q\mu_{\max}} - 1\bigg], 
\end{align}
where the last step is true for $\mu_{\max}\triangleq\max_{n}\mu_n$.

In order to lower-bound $\sum_{n=1}^{N}\bE[\wh{Q}_n]$, we construct a hypothetical single-server queue $\{\Phi[t]\}$ with the same arrival process $\{A[t]\}_{t\geq0}$ and an aggregated service process $\{R_{\Sigma}[t]\}_{t\geq0}$, where $R_{\Sigma}[t]\triangleq\sum_{n=1}^{N}R_n[t]$. The queue-length evolution of this single-server queue can be written as:
$\Phi[t+1]=\max\{\Phi[t]+A[t]-R_{\Sigma}[t],0\}$.
Due to resource pooling, the constructed hypothetical single-server's queue-length $\{\Phi[t]\}_{t\geq0}$ is stochastically smaller than $\{\sum_{n=1}^{N}Q_n[t]\}_{t\geq0}$ under any feasible policy. 
Hence, by \cite[Lemma 5]{eryilmaz2012asymptotically}, we immediately have the following lower bound:
\begin{align} \label{eqn:QueueLB_asym}
\sum_{n=1}^{N}\bE[\wh{Q}_n]\geq \frac{MN}{\epsilon},
\end{align}
where $M\triangleq\frac{\epsilon}{2N(\mu_{\Sigma}-\lambda)}\big(\text{Var}(A[t])+\sum_{n=1}^{N}\text{Var}(R_n[t])+(\mu_{\Sigma}-\lambda)^2\big)-\frac{1}{2}\epsilon R_{\max}$.
Lastly, combining (\ref{eqn:AgeLB_asym}),  (\ref{eqn:QueueLB_asym}), and \eqref{eqn:prop:stochastoc:ub} yields the desired result in Theorem~\ref{thm:general:asym} and the proof is complete.
\end{proof}

\begin{rem} \label{rem2}
{\em
From Theorem \ref{thm:general:asym}, we can see that for any fixed $\gamma$ value, we have
\begin{align*}
\lim_{\beta\rightarrow\infty}\rho(\beta,\gamma)\leq1-\frac{1}{2}\frac{\frac{1}{q\mu_{\max}}\sum_{n=1}^{N}\mu_n-1}{\frac{N}{q}-1},
\end{align*}
whose upper bound is equal to $1/2$ in the case with symmetric services, i.e., $\mu_1=\mu_2=\cdots=\mu_N$. However, we shall see from the numerical results presented in Section \ref{sec:numerical} that for any fixed $\gamma$ value, as $\beta$ increases, the PoA actually converges to zero in the case with symmetric services. 
The looseness of the upper bound analysis is due to the intrinsic nature of the Lyapunov analysis methodology, which only captures the drift among neighboring slots in temporal domain and does not characterize the Round-Robin behavior in spatial domain. }
\end{rem}


\begin{figure*}[t!]
    \begin{minipage}[t]{0.32\linewidth}
        \centering
        \includegraphics[width=1.1\textwidth]{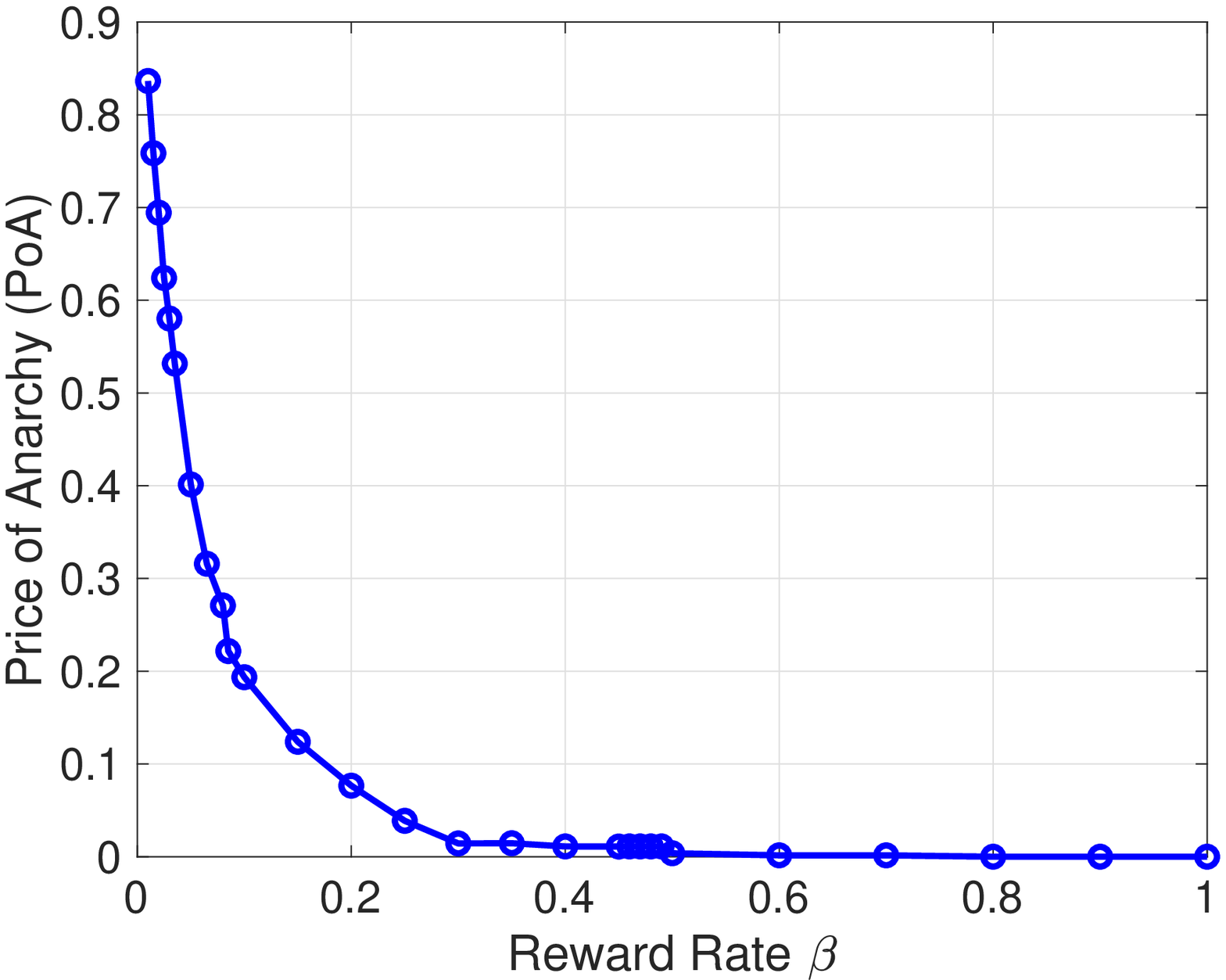}
        \caption{PoA with respect to reward rate $\beta$ in the deterministic case.} \label{fig:sim:periodic}
    \end{minipage}%
    \hspace{0.009\textwidth}
    \begin{minipage}[t]{0.32\linewidth}
        \centering
        \includegraphics[width=1.1\textwidth]{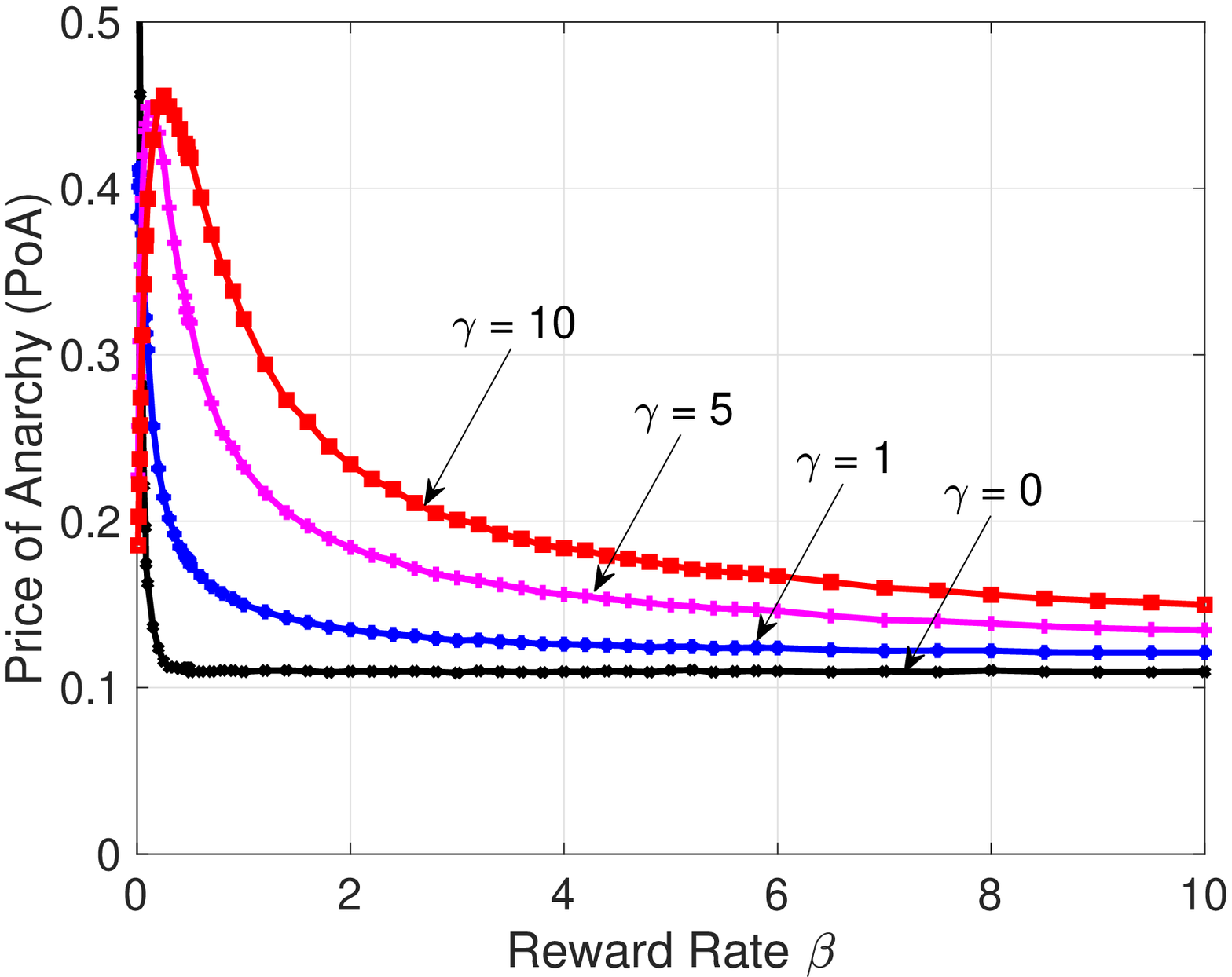}
        \caption{PoA with respect to $\beta$ and $\gamma$ in the stochastic case with asymmetric services.} \label{fig:general_asym_imperfect}
    \end{minipage}%
    \hspace{0.009\linewidth}
    \begin{minipage}[t]{0.32\linewidth}
        \centering
        \includegraphics[width=1.1\textwidth]{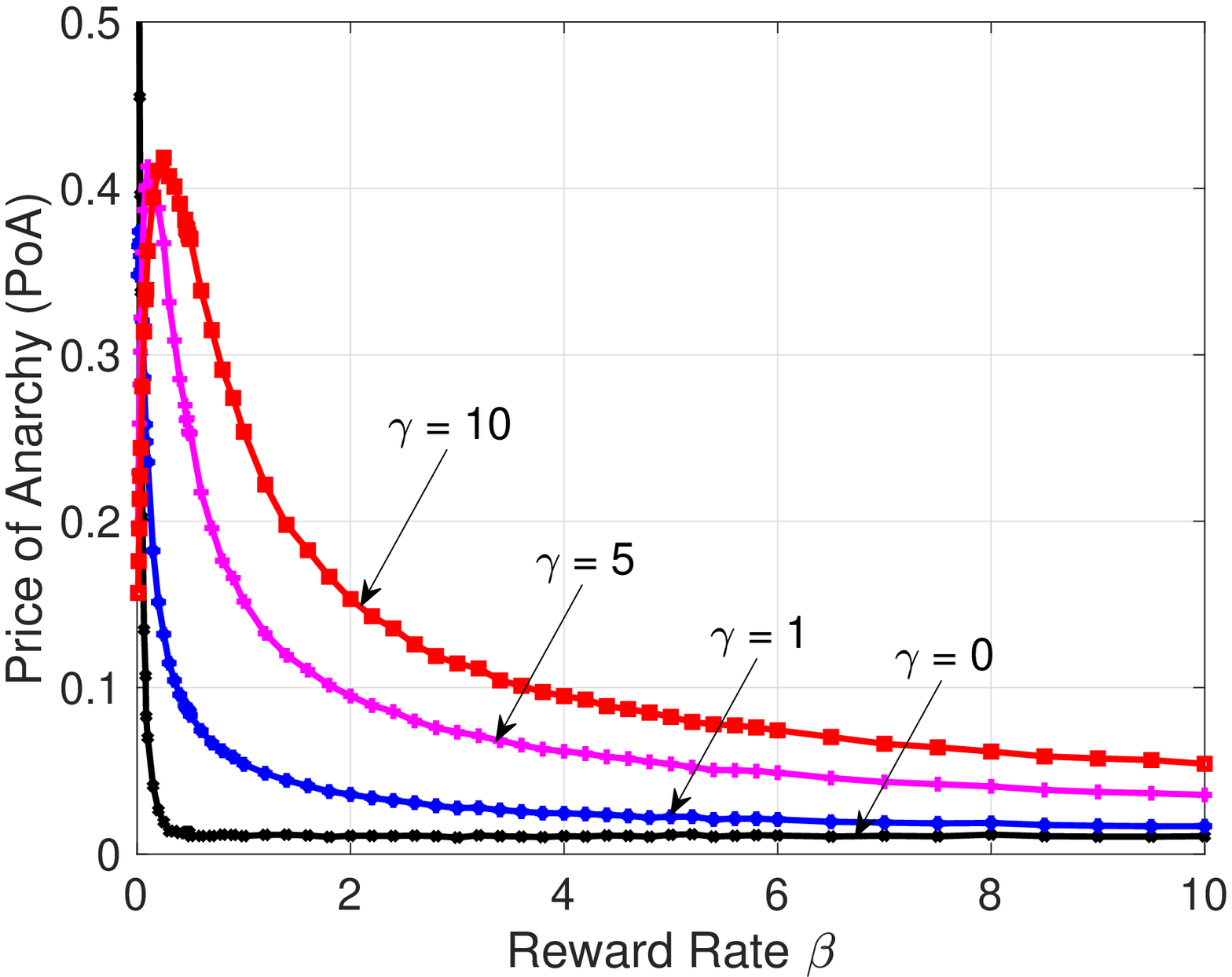}
        \caption{PoA with respect to $\beta$ and $\gamma$ in the stochastic case with symmetric services.}\label{fig:general_asym_perfect}
    \end{minipage}%
\vspace{-.2in}
\end{figure*}

\section{Numerical Results} \label{sec:numerical}

In this section, we conduct simulations to study the PoA performance under users' selfishness (cf. \eqref{eqn:user:behavior:general}) in a mobile crowed-learning system. 
We use a 10-PoI system and assume that each PoI $n$'s state information $p_n[t]$ belongs to the finite set $\{0.25,0.5,0.75,1\}$, and $p_{n}[t]$ changes to a different value uniformly at random every $100$ time slots. 
We consider both deterministic and stochastic cases. 
For the deterministic case, we assume that there is exactly one arriving user in each time slot and each PoI can serve one user in one time slot if any. 
For the stochastic case, we assume that users arrive at the system according to the Bernoulli distribution with mean $\lambda=0.9$ and service provided by each PoI $n$ follows an i.i.d. Bernoulli distribution with mean $\mu_n$, $n=1,2,\ldots,10$. 
We consider both symmetric and asymmetric services:
For symmetric services, we let $\mu_n=0.1, \forall n$;
For asymmetric services, we let $\mu_n=0.11,\forall n=1,2,\ldots,5$ and $\mu_n=0.09, \forall n=6,7,\ldots,10$. 

\smallskip
{\em 1) Deterministic Scenario:}
Fig. \ref{fig:sim:periodic} illustrates the PoA performance in the deterministic case. 
In this case, there is no queueing effect and the PoA performance reflects the information freshness due to users' selfish behavior compared to the optimal AoI performance. 
We can observe from Fig.~\ref{fig:sim:periodic} that PoA decreases as the reward rate $\beta$ increases and roughly follows the $O(1/\beta)$ law, meaning that the AoI performance improves. Moreover, PoA decreases to zero for $\beta\geq 0.5$. 
This means that the AoI performance is optimal even with selfish users. 
Both observations corroborate the result in Theorem \ref{thm:periodic}. 

\smallskip
{\em 2) Stochastic Scenario:}
Next, we study the PoA performance in stochastic cases.
We consider both symmetric and asymmetric services.
Here, PoA reflects the gap between joint AoI-congestion performance under users' selfishness compared to the optimal performance.
We note that, even without incorporating AoI, it remains an open problem to find an optimal policy to minimize the total mean queue-length. 
In deriving the upper bound on PoA, we use the fundamental lower bound on total mean queue-length (cf. \cite[Lemma~5]{eryilmaz2012asymptotically}), which may not be tight. 
In this simulation, we adopt the Join-the-Shortest-Queue (JSQ) policy (e.g., \cite{eryilmaz2012asymptotically}) and use its mean queue-length to serve as a lower bound for the queueing component in PoA. 
This is because JSQ minimizes the total mean queue-length (see \cite[Proposition~3]{li2013optimaloverload}) in the case with Bernoulli arrival and symmetric Bernoulli services, and it is optimal (see \cite{eryilmaz2012asymptotically}) in the case with general arrival and service processes in the heavy-traffic regime (i.e., arrival rate approaches the total service rate asymptotically). 



Fig.~\ref{fig:general_asym_imperfect} shows the PoA performance of the case with asymmetric services under different values of $\gamma$. 
We can see that, for any fixed $\gamma$ value, PoA converges to $0.1$ instead of $0$ as $\beta$ increases. 
The main reason is that we adopt the weighted sum of mean age as the metric for information freshness, and the policy that achieves optimal information freshness is unknown.
Thus, we use our derived fundamental lower bound on the weighted mean sum-age to replace the optimal value for the information freshness, which may render a loose bound on PoA. 
However, we point out that our derived lower bound is tight in the symmetric service case as the reward rate $\beta$ increases asymptotically, even though the derived upper bound of PoA is $1/2$ (cf. Remark~\ref{rem2}). 
Indeed, we can observe from Fig.~\ref{fig:general_asym_perfect} that PoA actually converges to zero as $\beta$ increases in the case with symmetric services. 


\section{Conclusion} \label{sec:conclusion}

In this paper, we have strived to understand whether or not we can achieve information freshness guarantee with selfish users in mobile crowd-learning.
To answer this question, we first developed a new analytical model that takes into account the essential features of mobile crowd-learning.
Then, based on this model, we showed that the natural greedy behavior of selfish users could lead to AoI instability, which necessitates the design of reward mechanisms to induce information freshness guarantee.
Toward this end, we proposed a linear AoI-based reward mechanism, under which we analyzed the impacts of users' selfishness on AoI based on the notion of Price of Anarchy (PoA).
We showed that the proposed reward mechanism achieves bounded AoI and congestion performances in terms of PoA, and can even achieves optimal AoI asymptotically in a deterministic scenario.
Collectively, these results serve as an exciting first step toward optimizing information freshness in mobile crowd-learning systems.

\bibliographystyle{IEEEtran}
\bibliography{IEEEabrv,./BIB/refs,./BIB/Crowdsensing,./BIB/AOI,./BIB/Bin}

\begin{thebibliography}{10}
\providecommand{\url}[1]{#1}
\csname url@samestyle\endcsname
\providecommand{\newblock}{\relax}
\providecommand{\bibinfo}[2]{#2}
\providecommand{\BIBentrySTDinterwordspacing}{\spaceskip=0pt\relax}
\providecommand{\BIBentryALTinterwordstretchfactor}{4}
\providecommand{\BIBentryALTinterwordspacing}{\spaceskip=\fontdimen2\font plus
\BIBentryALTinterwordstretchfactor\fontdimen3\font minus
  \fontdimen4\font\relax}
\providecommand{\BIBforeignlanguage}[2]{{%
\expandafter\ifx\csname l@#1\endcsname\relax
\typeout{** WARNING: IEEEtran.bst: No hyphenation pattern has been}%
\typeout{** loaded for the language `#1'. Using the pattern for}%
\typeout{** the default language instead.}%
\else
\language=\csname l@#1\endcsname
\fi
#2}}
\providecommand{\BIBdecl}{\relax}
\BIBdecl

\bibitem{Waze}
\BIBentryALTinterwordspacing
{Waze Mobile App}. [Online]. Available: \url{https://www.waze.com/}
\BIBentrySTDinterwordspacing

\bibitem{GasBuddy}
\BIBentryALTinterwordspacing
{GasBuddy Mobile App}. [Online]. Available: \url{https://www.gasbuddy.com/}
\BIBentrySTDinterwordspacing

\bibitem{Pavemint}
\BIBentryALTinterwordspacing
{Pavemint Mobile App}. [Online]. Available: \url{https://www.pavemint.com/}
\BIBentrySTDinterwordspacing

\bibitem{WiFiFinder}
\BIBentryALTinterwordspacing
{WiFi Finder Connect Internet, Mobile App}. [Online]. Available:
  \url{https://itunes.apple.com/us/app/wifi-finder-connect-internet/id1011519183?mt=8}
\BIBentrySTDinterwordspacing

\bibitem{Basket}
\BIBentryALTinterwordspacing
{Basket -- Grocery Shopping Mobile App}. [Online]. Available:
  \url{https://itunes.apple.com/us/app/basket-grocery-shopping/id1060139875?mt=8}
\BIBentrySTDinterwordspacing

\bibitem{Kaul12:AO_CISS}
S.~Kaul, R.~D. Yates, and M.~Gruteser, ``Status updates through queues,'' in
  \emph{Proc. CISS}, Princeton, NJ, USA, March 2012, pp. 1--6.

\bibitem{Yates12:AO_ISIT}
R.~D. Yates and S.~Kaul, ``Real-time status updating: Multiple sources,'' in
  \emph{Proc. IEEE ISIT}, Cambridge, MA, USA, July 2012, pp. 2666--2670.

\bibitem{Kaul12:AO_INFOCOM}
S.~Kaul, R.~D. Yates, and M.~Gruteser, ``Real-time status: How often should one
  update?'' in \emph{Proc. IEEE INFOCOM}, Orlando, FL, USA, March 2012, pp.
  2731--2735.

\bibitem{Kaul11:AO_SECON}
S.~Kaul, M.~Gruteser, V.~Rai, and J.~Kenney, ``Minimizing age of information in
  vehicular networks,'' in \emph{Proc. IEEE SECON}, Salt Lake City, UT, USA,
  June 2011, pp. 350--358.

\bibitem{Creek_watch}
\BIBentryALTinterwordspacing
``Creek watch,'' 2010. [Online]. Available:
  \url{http://creekwatch.researchlabs.ibm.com/}
\BIBentrySTDinterwordspacing

\bibitem{Maisonneuve09:noise_pollution}
N.~Maisonneuve, M.~Stevens, M.~E. Niessen, and L.~Steels, ``Noisetube:
  Measuring and mapping noise pollution with mobile phones.''\hskip 1em plus
  0.5em minus 0.4em\relax Springer, 2009, pp. 215--228.

\bibitem{Rana10:noise_mapping}
R.~K. Rana, C.~T. Chou, S.~S. Kanhere, N.~Bulusu, and W.~Hu, ``Ear-phone: an
  end-to-end participatory urban noise mapping system,'' in \emph{Proceedings
  of the 9th ACM/IEEE International Conference on Information Processing in
  Sensor Networks}.\hskip 1em plus 0.5em minus 0.4em\relax ACM, 2010, pp.
  105--116.

\bibitem{Zhang14:crowdsensing}
D.~Zhang, H.~Xiong, L.~Wang, and G.~Chen, ``Crowdrecruiter: selecting
  participants for piggyback crowdsensing under probabilistic coverage
  constraint,'' in \emph{Proceedings of the 2014 ACM International Joint
  Conference on Pervasive and Ubiquitous Computing}.\hskip 1em plus 0.5em minus
  0.4em\relax ACM, 2014, pp. 703--714.

\bibitem{Liu16:Crowdsensing_Survey}
J.~Liu, H.~Shen, and X.~Zhang, ``A survey of mobile crowdsensing techniques: A
  critical component for the internet of things,'' in \emph{Proc. IEEE ICCCN},
  2016.

\bibitem{Chen16:search_system}
Y.~Chen, J.~Zhou, and M.~Guo, ``A context-aware search system for internet of
  things based on hierarchical context model.'' \emph{Telecommunication
  Systems}, vol.~62, no.~1, pp. 77--91, 2016.

\bibitem{Han16:Crowdsensing_Game}
K.~Han, H.~Huang, and J.~Luo, ``Posted pricing for robust crowdsensing,'' in
  \emph{Proc. ACM MobiHoc}, 2016, pp. 261 -- 270.

\bibitem{Sun17:RT_ISIT}
Y.~Sun, Y.~Polyanskiy, and E.~Uysal-Biyikoglu, ``Remote estimation of the
  wiener process over a channel with random delay,'' in \emph{Proc. IEEE ISIT},
  June 2017, pp. 321--325.

\bibitem{Gao15:RT_CDC}
X.~Gao, E.~Akyol, and T.~Basar, ``Optimal estimation with limited measurements
  and noisy communication,'' in \emph{Proc. IEEE CDC}, Osaka, Japan, December
  2015, pp. 1775--1780.

\bibitem{Ceran17:SCC_ArXiv}
\BIBentryALTinterwordspacing
E.~T. Ceran, D.~Gunduz, and A.~́s Gyorgy. (2017, October) Average age of
  information with hybrid arq under a resource constraint. [Online]. Available:
  \url{https://arxiv.org/abs/1710.04971v1}
\BIBentrySTDinterwordspacing

\bibitem{Yates17:SCC_ISIT}
R.~D. Yates, E.~Najm, E.~Soljanin, and J.~Zhong, ``Timely updates over an
  erasure channel,'' in \emph{Proc. IEEE ISIT}, 2017.

\bibitem{Kam17:CF_ISIT}
C.~Kam, S.~Kompella, G.~Nguyen, J.~Wieselthier, and A.~Ephremides,
  ``Information freshness and popularity in mobile caching,'' in \emph{Proc.
  IEEE ISIT}, June 2017, pp. 136--140.

\bibitem{Yates17:CF_ISIT}
R.~D. Yates, P.~Ciblat, A.~Yener, and M.~Wigger, ``Age-optimal constrained
  cache updating,'' in \emph{Proc. IEEE ISIT}, June 2017, pp. 141--145.

\bibitem{Sun17:AO_TIT}
Y.~Sun, E.~Uysal-Biyikoglu, R.~D. Yates, C.~E. Koksal, and N.~B. Shroff,
  ``Update or wait: How to keep your data fresh,'' \emph{IEEE Transactions on
  Information Theory}, vol.~63, no.~11, pp. 7492--7508, November 2017.

\bibitem{Kaul17:AO_ISIT}
S.~K. Kaul and R.~D. Yates, ``Status updates over unreliable multiaccess
  channels,'' in \emph{Proc. IEEE ISIT}, 2017, pp. 331--335.

\bibitem{Li15:AS_INFOCOM}
B.~Li, A.~Eryilmaz, and R.~Srikant, ``On the universality of age-based
  scheduling in wireless networks,'' in \emph{Proc. IEEE INFOCOM}, Kowloon,
  Hong Kong, April 2015, pp. 1302--1310.

\bibitem{Lakshminarayana09:AS_CHPCN}
N.~B. Lakshminarayana, J.~Lee, and H.~Kim, ``Age based scheduling for
  asymmetric multiprocessors,'' in \emph{Conference on High Performance
  Computing Networking, Storage and Analysis}, Portland, OR, USA, November
  2009.

\bibitem{eryilmaz2012asymptotically}
A.~Eryilmaz and R.~Srikant, ``Asymptotically tight steady-state queue length
  bounds implied by drift conditions,'' \emph{Queueing Systems}, vol.~72, no.
  3-4, pp. 311--359, 2012.

\bibitem{li2013optimaloverload}
B.~Li, A.~Eryilmaz, R.~Srikant, and L.~Tassiulas, ``On optimal routing in
  overloaded parallel queues,'' in \emph{Decision and Control (CDC), 2013 IEEE
  52nd Annual Conference on}.\hskip 1em plus 0.5em minus 0.4em\relax IEEE,
  2013, pp. 262--267.

\end{thebibliography}


\end{document}